\documentclass[conference,letterpaper]{IEEEtran}

\usepackage[latin9]{inputenc}

\IEEEoverridecommandlockouts

\usepackage{amsthm}
\usepackage{amssymb}
\usepackage{amsmath}
\usepackage{amssymb}
\usepackage{epsfig}
\usepackage{epsf}
\usepackage{subfigure}
\usepackage{graphicx}
\usepackage{cite}
\usepackage{url}
\usepackage{enumerate}
\usepackage[]{authblk}
\usepackage{comment}

\usepackage{xcolor}

\def\BibTeX{{\rm B\kern-.05em{\sc i\kern-.025em b}\kern-.08em
    T\kern-.1667em\lower.7ex\hbox{E}\kern-.125emX}}

\newtheorem{claim}{Claim}

\newtheorem{theorem}{Theorem}

\newcommand{\dl}{\delta}

\newcommand{\msf}{\mathsf}
\newcommand{\mrm}{\mathrm}

\newcommand{\lp}{\left(}
\newcommand{\rp}{\right)}

\newcommand{\lbp}{\left\{}
\newcommand{\rbp}{\right\}}

\begin{document}

\title{Distortion-Based Outer-Bounds for Channels\\ with Rate-Limited Feedback}

\author{Alireza~Vahid,~\IEEEmembership{Senior~Member,~IEEE}
\thanks{A.~Vahid is with the Department of Electrical Engineering, University of Colorado Denver, USA. Email: {\sffamily alireza.vahid@ucdenver.edu}.} 
}

\maketitle

\begin{abstract}
We present a new technique to obtain outer-bounds on the capacity region of networks with ultra low-rate feedback. We establish a connection between the achievable rates in the forward channel and the minimum distortion that can be attained over the feedback channel. 
\end{abstract}



\section{Introduction}
\label{Section:Introduction_RLBC}

The introduction of massive Machine-Type Communications (mMTC) challenges many assumptions took for granted in network information theory. One of the main challenges is the increased cost of learning. For instance, in current systems, it is well justified to assume free access to small control packets (\emph{e.g.}, ACK/NACK signals) when needed, as they are much smaller is size compared to payload packets. However, this is no longer the case in mMTC where payload and control packets will be of comparable sizes. To make matters more complicated, there is a growing concern about security attacks that aim to disrupt unprotected feedback channels.

We present new outer-bounds for networks with low-rate feedback. The outer-bound establishes a connection between rate-distortion theory and achievable rates in multi-terminal networks with feedback. Interestingly, we learn that the best use of the feedback channel may not be to minimize the error but rather the distortion in reconstructing channel information. More specifically, we first quantify how closely the channel state may be reconstructed using the rate-limited feedback link, and then, we define a space of indistinguishable channel realizations in which the outer-bound is optimized.

To illustrate the technique, we focus on the two-user broadcast erasure channel (BEC) and assume a sub-bit feedback link from \emph{only} one receiver, $\msf{Rx}_F$, while the other receiver, $\msf{Rx}_N$, does not share any information with the other nodes. Receiver $\msf{Rx}_f$ uses the rate-limited feedback link to causally provide its potentially encoded channel state information (CSI) to the other nodes. Although it remains open whether the bound is tight, it is the first of its kind and the surprising message is that to achieve the bound, the approach may be to purposefully induce distortion into the encoded CSI. We outline how existing results fall short of achieving these bounds and provide further insights and interpretations. 


\noindent {\bf Related Literature:} 
In~\cite{ardestanizadeh2009wiretap}, outer-bounds for wiretap channels with rate-limited output feedback were derived, which are tight for physically degraded channels. 
Shayevitz and Wigger observed in~\cite{shayevitz2012capacity} that finding a general feedback capacity formula for memoryless broadcast channels (BCs) is very hard. For other channels such as Gaussian
BCs, the capacity with single-user feedback is still unknown~\cite{amor2015mimo}. The block Markov coding of~\cite{wu2016coding} provides interesting insights but involves characterizing complicated auxiliary random variables.
To further understand how low-rate feedback affects the capacity region of multi-terminal channels, two-user broadcast erasure channels with intermittent~\cite{vahid2019capacity,vahid2020erasure} and one-sided~\cite{he2017two,lin2018gaussian,lin2019no} delayed feedback have been studied. Interestingly, it was shown in~\cite{lin2019no,JournalsingleuserCSI} that even when only one receiver provides its delayed CSI to the transmitter, the outer-bound with global delayed feedback can be achieved. This latter finding motivates us to further lower the feedback rate to sub-bit territory to understand the fundamental limits of communications with rate-limited feedback in its purest form.
In the context of interference channels, in~\cite{AlirezaISIT,vahid2012interference}, the capacity of two-user interference channels with rate-limited feedback was established; new coding schemes with noisy or intermittent output feedback were proposed in~\cite{gastpar2013coding,karakus2015gaussian}; and~\cite{tuninetti2012outer,cheng2013two,suh2018two} generalize such ideas to two-way communications. Finally, locality of feedback was studied under different delay assumptions in~\cite{vahid2010capacity,aggarwal2011achieving,vahid2015impact,vahid2016two,vahid2016does,vahid2017interference,vahid2019throughput}.




\section{Problem Formulation}
\label{Section:Problem_RLBC}

We consider the two-user broadcast erasure channel (BEC) of Fig.~\ref{Fig:BC_CEE} in which a single-antenna transmitter, $\msf{Tx}$, wishes to communicate two independent messages, $W_F$ and $W_N$, to two single-antenna receiving terminals $\msf{Rx}_F$ and $\msf{Rx}_N$ (read feedback/silent receiver), respectively, over $n$ channel uses. Each of the messages, $W_F$ and $W_N$, is uniformly distributed over $\lbp1,2,\ldots,2^{nR_F}\rbp$ and $\lbp1,2,\ldots,2^{nR_N}\rbp$, respectively. At time instant $t$, the messages are mapped to channel input $X[t] \in \mathbb{F}_2$ (in the binary field), and the respective received signals at $\msf{Rx}_F$ and $\msf{Rx}_N$ are:
\begin{align}
\label{eq_DL_channel}
Y_F[t] = S_F[t] X[t]~~ \; \mbox{and} \;~~ Y_N[t] = S_N[t] X[t], 
\end{align}
where $\lbp S_F[t]\rbp$ is the Bernoulli $(1-\delta_F)$ process that governs the erasure at $\mathsf{Rx}_F$, and $\lbp S_N[t]\rbp$ is the Bernoulli $(1-\delta_N)$ process that governs the erasure at $\mathsf{Rx}_N$. In this manuscript, we assume the channels are distributed independently over time and across users, and we limit the scope to $\delta_F = \delta_N = \delta$ (\emph{i.e.} homogeneous channels) where $\delta$ is known globally. 

We assume that at time instant $t$, each receiver knows its channel value instantly, \emph{e.g.}, at time instant $t$, $\msf{Rx}_F$ knows the realization of $S_F[t]$. When the channel realization is $1$, the corresponding terminal receives $X[t]$ noiselessly, and when it is $0$, the terminal understands an erasure has occurred. 

\begin{figure}[!ht]
\centering
\includegraphics[height = .5\columnwidth]{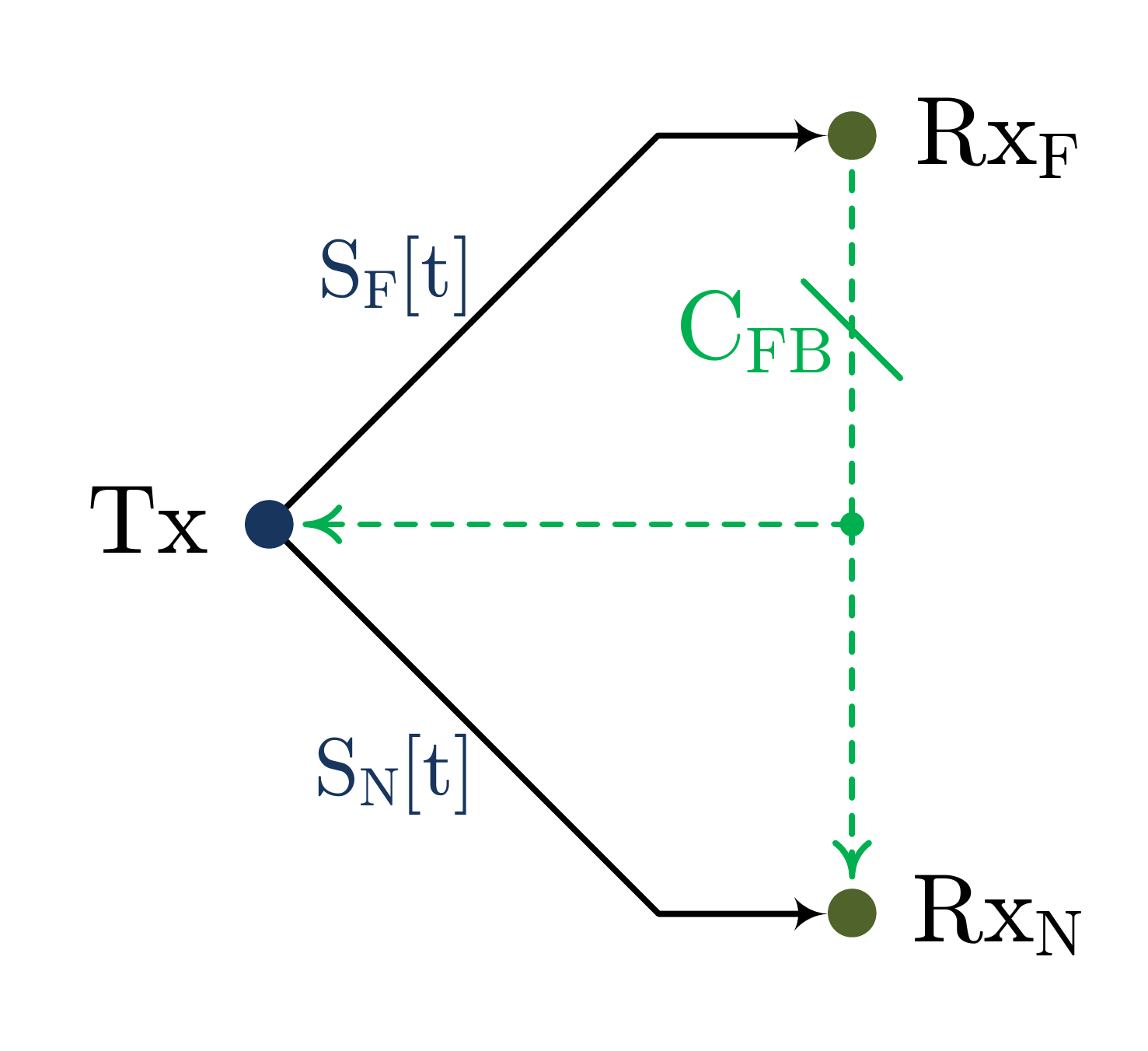}
\caption{Two-user BEC with rate-limited one-sided feedback.\label{Fig:BC_CEE}}
\end{figure}

We further assume a one-sided feedback structure in which at time instant $t$, $\mathsf{Rx}_N$ knows $S_N[t]$ but does \emph{not} share this information with the other nodes. On the other hand, we assume $\mathsf{Rx}_F$ shares its channel state with receiver $\mathsf{Rx}_N$ and the transmitter through a rate-limited feedback channel of capacity $C_{\sf FB}$.
We assume the feedback encoder at $\mathsf{Rx}_F$ sends back feedback symbol $K[t] \in \mathcal{K}[t]$ at time $t$ and this information becomes available to the other nodes at time $t+1$. Here, $\mathcal{K}[t]$ is the feedback alphabet at time $t$. The feedback symbol depends causally on $S_F[t]$, and the cardinality of the feedback alphabets satisfies
\begin{align}
\frac{1}{t}\sum_{\ell = 1}^{t}{\log_2\lp \left| \mathcal{K}[\ell] \right| \rp} \leq C_{\sf FB}, \qquad \forall~1 \leq t \leq n.
\end{align}

The constraint imposed at the encoding function $f_t(.) $ at time index $t$ is
\begin{align} 
\label{eq_enc_function}
X[t] = f_t\lp W_F, W_N, K^{t-1} \rp.
\end{align}

Receivers $\msf{Rx}_F$ and $\msf{Rx}_N$ use decoding functions $\varphi_{F,n}\left( Y_F^n, S_F^n \right)$ and $\varphi_{N,n}\left( Y_2^n, K^n, S_N^n \right)$ to get estimates $\widehat{W}_F$ of $W_F$ and $\widehat{W}_N$ of $W_N$, respectively. An error occurs whenever the estimate does not match the corresponding message. The average probabilities of error are given by 
\begin{align}
\lambda_{F,n} = \mathbb{E}[P(\widehat{W}_F \neq W_F)],~\lambda_{N,n} = \mathbb{E}[P(\widehat{W}_N \neq W_N)], 
\end{align}
where the expectations are taken with respect to the random choice of the transmitted messages.

We say that a rate pair $(R_F,R_N)$ is achievable if there exists a block encoder at the transmitter, and a block decoder at each receiver, such that the average probabilities of error go to zero as the block length $n$ goes to infinity. The capacity region, $\mathcal{C}$, is the closure of the set of achievable rate pairs. 



\section{Main Results}
\label{Section:Main_RLBC}

In this section, we provide the new distortion-based outer-bound for the two-user BEC with rate-limited feedback. Define $0 \leq D^\ast \leq \min\lbp \dl, 1-\dl \rbp$ to be the unique value to satisfy
\begin{align}
\label{Eq:MinDistortion}
H\lp D \rp = \lp H\lp \delta \rp -  C_{\sf FB} \rp^+,
\end{align}
and $\gamma_{\mrm{out}}$ to be
\begin{align}
\label{Def:Gamma}
\gamma_{\mrm{out}} \overset{\triangle}= \frac{D^\ast}{\min \lbp \dl, 1-\dl \rbp}.
\end{align}


\begin{theorem}
\label{THM:Out_RLBC}
The capacity region, $\mathcal{C}$, of the two-user BEC with one-sided rate-limited feedback is included in
\begin{equation}
\label{Eq:Cout-RLBC}
\mathcal{C}_{\mrm{out}} \overset{\triangle}=
\left\{ \begin{array}{ll}
\hspace{-1.5mm} \left( R_F, R_N \right) \left| \parbox[c][3em][c]{0.25\textwidth} {
$R_F + \beta_{\mrm{out}} R_N \leq \beta_{\mrm{out}} \left( 1 - \delta \right)$\\
$\beta_{\mrm{out}} R_F + R_N \leq \beta_{\mrm{out}} \left( 1 - \delta \right)$
} \right. \end{array} \right\},
\end{equation}
where
\begin{align}
\label{Eq:Beta_RLBC}
\beta_{\mrm{out}} &= \gamma_{\mrm{out}} + \lp 1 - \gamma_{\mrm{out}} \rp \lp 1 + \dl \rp. 
\end{align}
\end{theorem}

Similar to the findings of~\cite{lin2019no}, although only $\msf{Rx}_F$ provides feedback, the outer-bounds are symmetric. The results could also provide new capacity bounds for erasure interference channels based on their connection to BECs~\cite{vahid2012binary,AlirezaBFICDelayed,AlirezaInfocom2014,vahid2018arq}.

Figure~\ref{Fig:SumRateRLBC} plots the outer-bound on the maximum sum-rate point for $\delta = 0.4$.
When $C_{\sf FB} \geq H\lp \dl \rp$, this outer-bound matches the capacity region of the two-user BEC with global delayed feedback~\cite{gatzianas2013multiuser}. This recovers the results of~\cite{lin2019no,JournalsingleuserCSI} where it is shown that perfect one-sided feedback is as good as global feedback. The interesting distinction between the results in~\cite{lin2019no,JournalsingleuserCSI} and prior results is the fact that delayed CSI is harnessed at every step of the achievability. One can envision a Markov block structure to interleave different blocks and compress the feedback to $H\lp \dl \rp n$ bits and mimic the results of~\cite{lin2019no} to achieve of point ${\sf A}$ in Figure~\ref{Fig:SumRateRLBC} where $H\lp \delta \rp =  C_{\sf FB}$. However, below this limit, the achievability can no longer be derived from such arguments. In fact, the argument presented above requires perfect delayed CSI, and for $C_{\sf FB} < H\lp \dl \rp$, if we attempt to send perfectly a part of CSI back to the transmitter, we deviate from this outer-bound as in  Figure~\ref{Fig:SumRateRLBC}. This suggests the possibility of new achievability ideas that purposefully induce distortion in the encoded CSI.

\begin{figure}[!t]
\centering
\includegraphics[width = \columnwidth]{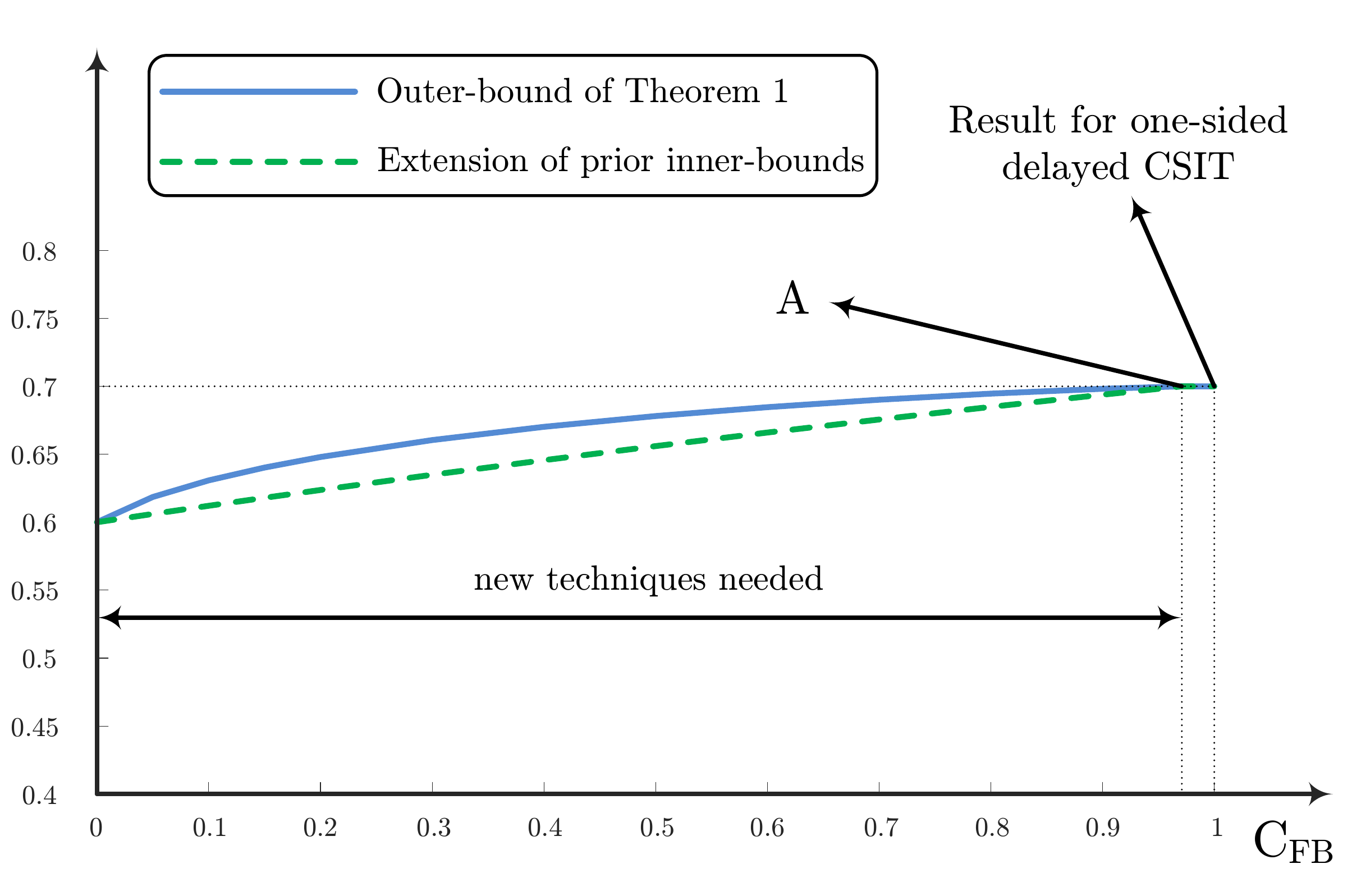}
\caption{Sum-rate outer-bound for $\delta = 0.4$. Point $A$ is where $H\lp \delta \rp =  C_{\sf FB}$.\label{Fig:SumRateRLBC}}
\vspace{-3mm}
\end{figure}

\section{Proof of Theorem~\ref{THM:Out_RLBC}}
\label{Section:Converse_RLBC}

\noindent {\bf Overview}: The proof is based on the following three key ingredients: (1) for $\msf{Rx}_F$, we find a set of channels realizations that result in the same feedback sequence as the original channel; (2) within this set, we find a candidate that has the same marginal distribution as $S_F^n$ but is maximally correlated with $S_N^n$; (3) as further discussed later, we use the fact that: 
\begin{align}
H(X^n|S_F^n, K^n) = H(X^n|K^n).
\end{align}

\noindent {\bf Derivation}: To derive the outer-bounds, we enhance the knowledge at $\msf{Rx}_N$ by providing it with $S_F^{t-1}$ rather than $K^{t-1}$. The transmitter is still informed causally of the CSI associated with $\msf{Rx}_F$ through a rate-limited feedback link of capacity $C_{\sf FB}$. The proof is broken into several steps as discussed below.

\noindent {\bf Step~1}: From Rate-Distortion Theory, we know that for a binary source distributed as i.i.d. $\mathcal{B}\lp 1-\dl \rp$, given a channel capacity of $C_{\sf FB}$, the minimum attainable Hamming distortion, $0 \leq D^\ast \leq \min\lbp \dl, 1-\dl \rbp$, between the source $S_F^n$ and the estimate $\hat{S}_F^n$, is the solution to~\eqref{Eq:MinDistortion}.


Fix the feedback strategy mapping $S_F^n$ to $K^n$, and suppose this strategy results in distortion, $D_{\sf FB} \geq D^\ast$, between $S_F^n$ and the estimate $\hat{S}_F^n$. We further denote the fraction of time instants in which $S_F[t] = 1$ and $\hat{S}_F[t] = 0$ by $D_{10}$. Similarly,  $D_{01}$ is the fraction of time instants in which $S_F[t] = 0$ and $\hat{S}_F[t] = 1$. Thus, we have
\begin{align}
D_{10} + D_{01} =  D_{\sf FB}.
\end{align}

\begin{claim}
\label{Claim:Conditions}
There exists $\tilde{S}_F^n$ such that: $(1)$ we have
\begin{align}
\label{Eq:TildeConditions}
\tilde{S}_F[t] \overset{i.i.d.}\sim \mathcal{B}\lp 1-\dl \rp \text{~and~} \mathbb{E}\lp d\lp \hat{S}_F^n, \tilde{S}_F^n \rp \rp  \leq D_{\sf FB},
\end{align}
where $d\lp \cdot, \cdot \rp$ is the Hamming distance; $(2)$ the feedback encoder maps $\tilde{S}_F^n$ to the same feedback sequence $K^n$ for $S_F^n$. 
\end{claim}

\begin{proof}
Suppose no such sequence exists. Then, the transmitter could have an estimate of $S_F^n$ with a distortion smaller than $D_{\sf FB}$, which would be in contradiction with Rate-Distortion Theory, thus, proving the result. 
\end{proof}

As noted, $\tilde{S}_F[t] \overset{i.i.d.}\sim \mathcal{B}\lp 1-\dl \rp$. Denote the fraction of time instants in which $\tilde{S}_F[t] = 1$ and $\hat{S}_F[t] = 0$  by $\tilde{D}_{10}$, and  the fraction of time instants in which $\tilde{S}_F[t] = 0$ and $\hat{S}_F[t] = 1$ by $\tilde{D}_{01}$. We have
$\tilde{D}_{10} + \tilde{D}_{01} \leq D_{\sf FB}$.
 
\noindent {\bf Step~2}: In this step, we claim if we replace sequence $S_F^n$ with $\tilde{S}_F^n$, the capacity region remains unchanged.

\begin{claim}
\label{Claim:ModifiedBC_RLBC}
Any achievable $\lp R_F, R_N \rp$ is included in the capacity region of a BEC with $\tilde{S}_F^n$ instead of $S_F^n$ while other parameters are kept the same and vice versa.
\end{claim}

\noindent \emph{Proof.} $I\lp W_F; Y_F^n | S_F^n \rp  = I\lp W_F; Y_F^n | S_F^n, K^n \rp$
\begin{align}
= I\lp W_F; \tilde{Y}_F^n | \tilde{S}_F^n, K^n \rp = I\lp W_F; \tilde{Y}_F^n | \tilde{S}_F^n \rp,
\end{align}
and $I\lp W_N; Y_N^n | S_F^n, S_N^n \rp = I\lp W_N; Y_N^n | S_F^n, S_N^n , K^n \rp$
\begin{align}
& = I\lp W_N; Y_N^n | S_N^n , K^n \rp = I\lp W_N; Y_N^n | \tilde{S}_F^n, S_N^n , K^n \rp.
\end{align}

\noindent {\bf Step~3}: In the rate-limited broadcast channel of Section~\ref{Section:Problem_RLBC}, $S_F[t]$ and $S_N[t]$ are distributed as independent Bernoulli random variables, and the feedback encoder at $\msf{Rx}_F$ is unaware of $S_N^n$. In this step, we create a ``worst-case'' scenario by creating maximum correlation between $S_N^n$ and $\tilde{S}_F^n$.

\begin{claim}
Under the conditions expressed in Claim~\ref{Claim:Conditions}, we have
\begin{align}
\label{Eq:MaxOff}
&\max_{\tilde{D}_{10} + \tilde{D}_{01} \leq D_{\sf FB}} \Pr \lp \tilde{S}_F[t] = S_N[t] = 0 \rp \nonumber \\
&~=\left\{ \begin{array}{ll}
\vspace{1mm} \dl^2 + D_{\sf FB}\lp 1 - \dl \rp, & \min \lbp \dl, 1-\dl \rbp = \dl, \\ 
\dl^2 + D_{\sf FB}\dl, & \min \lbp \dl, 1-\dl \rbp = 1 - \dl. 
\end{array} \right.
\end{align}
\end{claim}

\begin{figure}[!ht]
\centering
\includegraphics[width = \columnwidth]{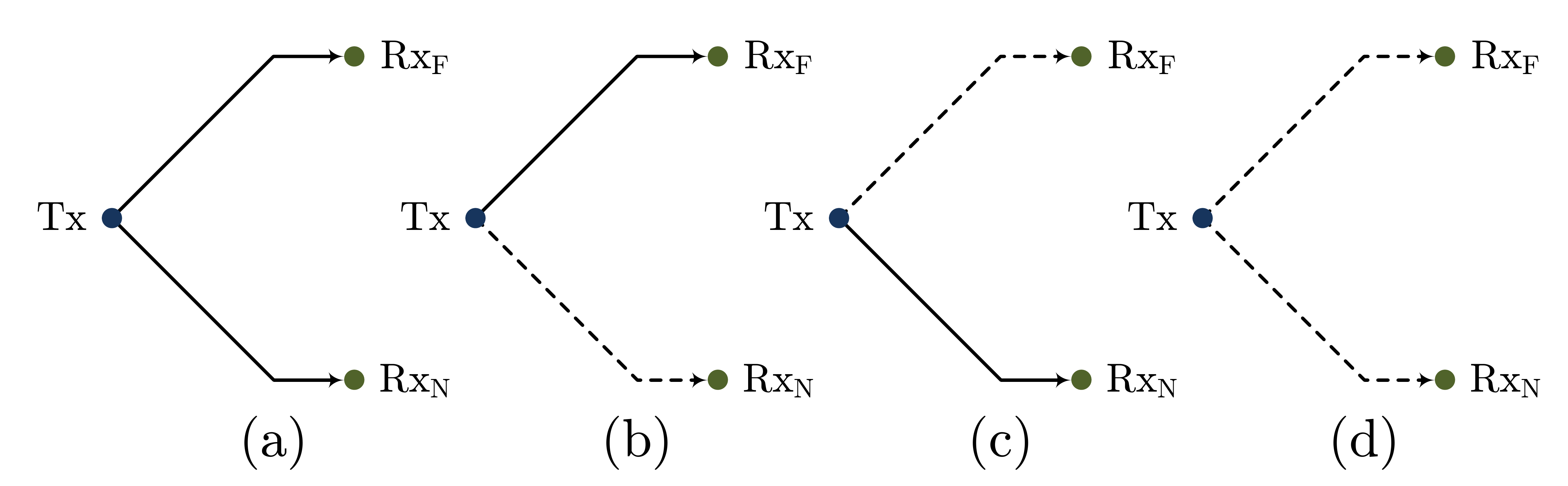}
\caption{Four possible channel realizations at each time.\label{Fig:Transition}}
\end{figure}

\begin{proof} We divide the proof into two parts based on $\dl$.

\noindent $\bullet$ ${\min \lbp \dl, 1-\dl \rbp = \dl:}$ In this case, to maximize $\Pr \lp \tilde{S}_F[t] = S_N[t] = 0 \rp$, the transition from $S_F^n$ to $\hat{S}_F^n$ would be such that to minimize the times in which $\hat{S}_F[t] = 1$ and $S_N[t] = 0$ by setting $D_{10}\lp 1 - \dl \rp n$ ones to zeros when $S_F[t] = 1$ and $S_N[t] = 0$. The transition from $\hat{S}_F^n$ to $\tilde{S}_F^n$ would be such that to maximize the times in which $\tilde{S}_F[t] = 1$ and $S_N[t] = 1$ by changing $D_{10}\lp 1 - \dl \rp n$ zeros in $\hat{S}_F^n$ to ones when $S_F[t] = 0$ and $S_N[t] = 1$. In effect, such transition from $S_F^n$ to $\tilde{S}_F^n$, increases the probability of $\tilde{S}_F[t] = S_N[t] = 0$ by $D_{10}\lp 1 - \dl \rp$, and since $D_{10} \leq D_{\sf FB}$, the maximum occurs at  $D_{10} = D_{\sf FB}$. Based on Figure~\ref{Fig:Transition}, the discussion above can be explained as follows. The transition from $S_F^n$ to $\hat{S}_F^n$ removes realizations from state $(b)$ and adds them to state $(d)$, while transition from $\hat{S}_F^n$ to $\tilde{S}_F^n$ moves realizations from state $d$ to $(c)$. Thus, end-to-end transitions are from state $(b)$ to $(c)$.  

\noindent $\bullet$ ${\min \lbp \dl, 1-\dl \rbp = 1 - \dl:}$ In this case, to maximize $\Pr \lp \tilde{S}_F[t] = S_N[t] = 0 \rp$, the transition from $S_F^n$ to $\hat{S}_F^n$ would be such that to maximize the times in which $S_F[t] = 1$ and $S_N[t] = 1$ by changing $D_{01} \dl n$ zeros to ones when $S_F[t] = 0$ and $S_N[t] = 1$. The transition from $\hat{S}_F^n$ to $\tilde{S}_F^n$ would be such that to minimize the times in which $\tilde{S}_F[t] = 1$ and $S_N[t] = 0$ by changing $D_{01}\dl n$ ones in $\hat{S}_F^n$ to zeros when $S_F[t] = 1$ and $S_N[t] = 0$. In effect, such transition from $S_F^n$ to $\tilde{S}_F^n$, increase the probability of $\tilde{S}_F[t] = S_N[t] = 0$ by $D_{01}\dl$, and since $D_{01} \leq D_{\sf FB}$, the maximum occurs at  $D_{01} = D_{\sf FB}$.
\end{proof}

To prepare for the final step, we define
\begin{align}
\tilde{\gamma}_{\mrm{out}} \overset{\triangle}= \frac{D_{\sf FB}}{\min \lbp \dl, 1-\dl \rbp}.
\end{align}
Based on this definition and \eqref{Eq:MaxOff}, we have
\begin{align}
&\min_{\tilde{D}_{10} + \tilde{D}_{01} \leq D_{\sf FB}} \Pr\left( \{ \tilde{S}_F[t] =  S_N[t] = 0 \}^{c} \right) \nonumber \\
&~= \tilde{\gamma}_{\mrm{out}} \lp 1 - \dl \rp + \lp 1 - \tilde{\gamma}_{\mrm{out}} \rp \lp 1 - \dl^2 \rp.
\end{align}
Further, from rate-distortion theory, we have $D_{\sf FB} \geq D^\ast$. Thus, for {\it any} feedback strategy, we have
\begin{align}
\label{Eq:MaxMin}
&A \overset{\triangle}=\max_{D_{\sf FB} \geq D^\ast} \min_{\tilde{D}_{10} + \tilde{D}_{01} \leq D_{\sf FB}} \Pr\left( \{ \tilde{S}_F[t] =  S_N[t] = 0 \}^{c} \right) \nonumber \\
&~= \gamma_{\mrm{out}} \lp 1 - \dl \rp + \lp 1 - \gamma_{\mrm{out}} \rp \lp 1 - \dl^2 \rp, 
\end{align}
where $\gamma_{\mrm{out}}$ is defined in \eqref{Def:Gamma}, and this bound is attained for a code that achieves the rate-distortion bound. In other words, for {\it any} feedback strategy, we have
\begin{align}
&\Pr\left( \{ \tilde{S}_F[t] =  S_N[t] = 0 \}^{c} \right) \nonumber \\
&~\leq \gamma_{\mrm{out}} \lp 1 - \dl \rp + \lp 1 - \gamma_{\mrm{out}} \rp \lp 1 - \dl^2 \rp
\end{align}
where $\tilde{S}_F^n$ satisfies the conditions in \eqref{Eq:TildeConditions}.

\noindent {\bf Step~4}: For the final step, we first prove the following result.

\begin{claim}
\label{Claim_beta_RLBC}
For the two-user BEC with rate-limited feedback as  described in Section~\ref{Section:Problem_RLBC} and for any input distribution, we have
\begin{align}
H\left( Y_F^n|W_N,S_F^n,S_N^n \right) - \beta_{\mrm{out}} H\left( Y_N^n|W_N,S_F^n,S_N^n \right) \leq 0,
\end{align}
where $\beta_{\mrm{out}}$ is given in \eqref{Eq:Beta_RLBC}.
\end{claim}


\begin{proof}
\begin{align}
&H\left( Y_N^n|W_N,S_F^n,S_N^n \right) \nonumber \\
& = H\left( Y_N^n|W_N,S_F^n,K^n,S_N^n \right) \overset{(a)}= H\left( Y_N^n|W_N,K^n,S_N^n \right) \nonumber \\
& = H\left( Y_N^n|W_N,\tilde{S}_F^n,K^n,S_N^n \right) + I\left( Y_N^n; \tilde{S}_F^n|W_N,K^n,S_N^n \right) \nonumber \\
& \overset{(b)}= H\left( Y_N^n|W_N,\tilde{S}_F^n,K^n,S_N^n \right) \overset{(c)}= H\left( Y_N^n|W_N,\tilde{S}_F^n,S_N^n \right) \nonumber \\
& \overset{(d)}= \sum_{t=1}^{n}{H\left( Y_N[t]|Y_N^{t-1},W_N,\tilde{S}_F^t,S_N^t \right)} \nonumber \\
& \overset{(e)}= \sum_{t=1}^{n}{(1-\dl)H\left( X[t]|Y_N^{t-1},W_N,\tilde{S}_F^t, S_N[t]=1,S_N^{t-1} \right)} \nonumber \\
& \overset{(f)}= \sum_{t=1}^{n}{(1-\delta)H\left( X[t]|Y_N^{t-1},W_N,\tilde{S}_F^{t},S_N^t \right)} \nonumber \\
& \overset{(g)}\ge \sum_{t=1}^{n}{(1-\delta)H\left( X[t]|\tilde{Y}_F^{t-1},Y_N^{t-1},W_N,\tilde{S}_F^{t},S_N^t \right)} \nonumber \\
& \overset{(h)} \ge \sum_{t=1}^{n}\frac{(1-\delta)}{A}H\left( \tilde{Y}_F[t], Y_N[t] |\tilde{Y}_F^{t-1},Y_N^{t-1},W_N,\tilde{S}_F^{t},S_N^t \right) \nonumber 
\end{align}
\begin{align}
& \overset{\eqref{Eq:MaxMin}} \ge \sum_{t=1}^{n}{\frac{(1-\delta)H\left( \tilde{Y}_F[t], Y_N[t] |\tilde{Y}_F^{t-1},Y_N^{t-1},W_N,\tilde{S}_F^{t},S_N^t \right)}{\gamma_{\mrm{out}} \lp 1 - \dl \rp + \lp 1 - \gamma_{\mrm{out}} \rp \lp 1 - \dl^2 \rp}} \nonumber \\
& \overset{(\ref{Eq:Beta_RLBC})}= \sum_{t=1}^{n}{\frac{1}{\beta_{\mrm{out}}}H\left( \tilde{Y}_F[t], Y_N[t] |\tilde{Y}_F^{t-1},Y_N^{t-1},W_N,\tilde{S}_F^{t},S_N^t \right)} \nonumber \\
& \overset{(i)}= \sum_{t=1}^{n}{\frac{1}{\beta_{\mrm{out}}}H\left( \tilde{Y}_F[t], Y_N[t] |\tilde{Y}_F^{t-1},Y_N^{t-1},W_N,\tilde{S}_F^{n},S_N^n \right)} \nonumber \\
&= \frac{H\left( \tilde{Y}_F^n, \tilde{Y}_N^n |W_N,\tilde{S}_F^{n},S_N^n \right)}{\beta_{\mrm{out}}} \overset{(j)}\ge \frac{H\left( \tilde{Y}_F^n|W_N,\tilde{S}_F^{n},S_N^n \right)}{\beta_{\mrm{out}}} \nonumber \\
&\overset{\mathrm{Claim}~\ref{Claim:ModifiedBC_RLBC}}\ge \frac{1}{\beta_{\mrm{out}}}H\left( Y_F^n|W_N,S_F^{n},S_N^n \right), 
\end{align}
where $(a)~\&~(b)$ hold since conditioned on $K^{t-1}$, $X[t]$ is independent of all other channel parameters; $(c)$ follows the fact that $\tilde{S}_F^n$ results in $K^n$; $(d)$ follows from the chain rule and the causality of the channel; $(e)$ holds since $S_N[t]$ is a Bernoulli $\left( 1 - \delta \right)$ process; $(f)$ is true since $X[t]$ is independent of channel realizations at time instant $t$; $(g)$ holds since conditioning reduces entropy; $(h)$ follows by the definition of $A$ in \eqref{Eq:MaxMin}, normalizing by the probability that at least one signal is not erased, and ensuring the inequality holds for any feedback strategy; $(i)$ follows the causality assumption; and $(j)$ holds as the discrete entropy function is non-negative.
\end{proof}

Finally, we are ready to prove the outer-bounds.
\begin{align}
&n \left( R_F + \beta_{\mrm{out}} R_N \right) = H\left(W_F\right) + \beta_{\mrm{out}} H\left(W_N\right) \\
& = H\left( W_F|W_N,S_F^n,S_N^n \right) + \beta_{\mrm{out}} H\left( W_N|S_F^n,S_N^n \right) \nonumber \\
& \overset{\mathrm{Fano}}\leq I\left( W_F; Y_F^n|W_N,S_F^n,S_N^n \right) \nonumber \\
&+ \beta_{\mrm{out}} I\left( W_N; Y_N^n|S_F^n,S_N^n \right) + n \epsilon_n \nonumber \\
& = H\left( Y_F^n|W_N,S_F^n,S_N^n \right) - \underbrace{H\left( Y_F^n|W_F,W_N,S_F^n,S_N^n \right)}_{=~0} \nonumber \\
&+ \beta_{\mrm{out}} H\left( Y_N^n|S_F^n,S_N^n \right) - \beta_{\mrm{out}} H\left( Y_N^n|W_N,S_F^n,S_N^n \right) + n \epsilon_n \nonumber \\
& \overset{\mathrm{Claim}~\ref{Claim_beta_RLBC}}\leq \beta_{\mrm{out}} H\left( Y_N^n|S_F^n,S_N^n \right) + n \epsilon_n \leq n \beta_{\mrm{out}} (1-\delta) + n \epsilon_n, \nonumber 
\end{align}
where $\epsilon_n \rightarrow 0$ as $n \rightarrow \infty$. Dividing both sides by $n$ and taking the limit for $n \rightarrow \infty$ completes the proof. The other outer-bound in Theorem~\ref{THM:Out_RLBC} can be obtained similarly.

\section*{Acknowledgement}

The research of A. Vahid was supported in part by NSF grant ECCS-2030285.

\bibliographystyle{ieeetr}
\bibliography{bib_RLBC.bib}

\end{document}